\documentclass[letterpaper, 10 pt, conference]{ieeeconf}
\IEEEoverridecommandlockouts
\usepackage{amsmath,amssymb,amsfonts}

\usepackage{amsthm}
\usepackage{mathtools}
\usepackage{algorithm,algorithmic}
\usepackage{booktabs}
\usepackage{cite}
\usepackage{hyperref}
\usepackage{xcolor}
\usepackage{graphicx}
\usepackage{tikz}
\usetikzlibrary{arrows.meta,positioning,shapes.geometric,fit,backgrounds}
\hypersetup{colorlinks=true,linkcolor=blue,citecolor=blue,urlcolor=blue}
\usepackage{url}

\newtheorem{theorem}{Theorem}

\newtheorem{assumption}{Assumption}
\newtheorem{remark}{Remark}
\newtheorem{definition}{Definition}
\newtheorem{problem}{Problem}

\newcommand{\R}{\mathbb{R}}
\newcommand{\Spp}{\mathbb{S}_{++}^n}
\newcommand{\E}{\mathbb{E}}
\newcommand{\tr}{\mathrm{tr}}
\newcommand{\Normal}{\mathcal{N}}
\newcommand{\ip}[2]{\left\langle #1,#2\right\rangle_F}
\newcommand{\Us}{\mathcal{U}^s}

\begin{document}

\title{Path Integral Control for Partially Observed Systems\\
with Controlled Sensing}

\author{Goutam Das, Takashi Tanaka
\thanks{This work is supported by DARPA COMPASS program grant HR0011-25-3-0210 and AFOSR DSCT program grant FA9550-25-1-0347.}
\thanks{
All authors are associated with the Networked Control Systems Lab at Purdue University. Emails: \{das347,~ttanaka16\}@purdue.edu.}}

\maketitle

\begin{abstract}
Path integral control in Gaussian belief space requires a structural matching condition between the observation-driven diffusion of the belief mean and the actuation authority, which a fixed observation matrix cannot enforce.
We treat the observation matrix as a control variable and show that constraining the sensing control to a measurable selector from the resulting \emph{matching set} reduces the Hamilton--Jacobi--Bellman equation for the belief mean and covariance to a linear PDE with a Feynman--Kac representation.
\end{abstract}

\section{Introduction}

\noindent Path integral control (PIC) \cite{kappen2005path,kappen2005linear,todorov2006linearly} solves a class of stochastic optimal control problems by linearizing the Hamilton--Jacobi--Bellman (HJB) equation via the Cole--Hopf transform, under a structural \emph{matching condition} between the diffusion of the controlled state and the control authority.
When matching holds, the value function admits a Feynman--Kac representation and is computable by forward Monte Carlo sampling; this principle underlies model predictive path integral control (MPPI) and its variants \cite{theodorou2010generalized,williams2017mppi,williams2018information,thijssen2015path,gandhi2021robust,yin2023risk}.
This framework, however, is focused on fully observed systems.
 
Under partial observability, the natural state is the belief over the physical state, whose evolution follows the Kushner--Stratonovich equation \cite{kushner1967dynamical,liptser1977statistics}.
Gaussian belief reductions via the Kalman--Bucy filter are standard for trajectory optimization \cite{bensoussan1992stochastic,krishnamurthy2016partially,platt2010belief,van2012motion,todorov2005generalized}.
Belief-space path integral control under a Gaussian belief approximation was recently developed in \cite{das2026cdc}, which established the Cole--Hopf linearization and Feynman--Kac representation for the belief mean when the observation matrix is fixed.
 
\emph{Problem:} If the observation matrix is itself a control variable, under what conditions does the resulting belief-space Hamilton--Jacobi--Bellman equation admit Cole--Hopf linearization and a Feynman--Kac representation?
 
We show that restricting the sensing control to a measurable selector from the \emph{matching set}---the set of sensing controls that enforce matching at the current covariance---yields a linear equation for the joint belief mean and covariance, whose solution admits a Feynman--Kac representation (Theorem~\ref{thm:main}).
The matching condition---previously determined by the system---becomes an active constraint that the controller enforces through sensor selection, extending the fixed-sensor result of \cite{das2026cdc}.

\begin{figure}[t]
\centering
\begin{tikzpicture}[
  >={Stealth[length=2mm]},
  thick,
  font=\small,
  block/.style = {draw, rectangle, minimum height=8mm,
                  minimum width=16mm, align=center, inner sep=2pt},
]
\node[block] (plant)  at (0,1.8)    {Plant};
\node[block] (sensor) at (2.5,1.8)  {$C(u^s_t)$};
\node[block] (filter) at (5.0,1.8)  {KBF};
 
\node[block, minimum width=65mm] (ctrl) at (2.5,0) {Controller};
 
\draw[->] (plant)  -- node[above] {$x_t$} (sensor);
\draw[->] (sensor) -- node[above] {$y_t$} (filter);
 
\draw[->] (ctrl.north -| plant.south)  -- node[right] {$u^a_t$}
          (plant.south);
\draw[->] (ctrl.north)                 -- node[right] {$u^s_t$}
          (sensor.south);
\draw[->] (filter.south)               -- node[right] {$(\mu_t,\Sigma_t)$}
          (ctrl.north -| filter.south);
\end{tikzpicture}
\caption{Illustration of the closed-loop system with controlled sensing.}
\label{fig:blockdiag}
\end{figure}
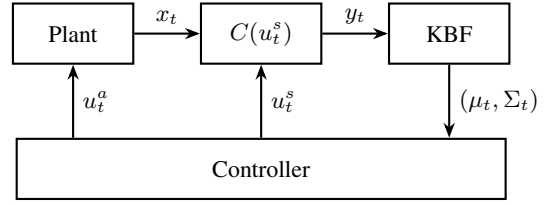

\section{Problem Formulation}
\label{sec:setup}

\noindent Consider the partially observed system
\begin{align}
dx_t &= f(x_t)\,dt + B u^a_t\,dt + H\,dw_t, \label{eq:state}\\
dy_t &= C(u^s_t)\,x_t\,dt + \sigma_o\,d\nu_t, \label{eq:obs}
\end{align}
with $x_t \in \R^n$, $u^a_t \in \R^{\ell_a}$, $u^s_t \in \R^{\ell_s}$, independent standard Brownian motions $w_t \in \R^m$ and $\nu_t \in \R^p$, $B \in \R^{n \times \ell_a}$, $H \in \R^{n \times m}$, $\sigma_o \in \R^{p \times p}$, $C(u^s) \in \R^{p \times n}$, and $Q := H H^\top$, $R_o := \sigma_o\sigma_o^\top \succ 0$.
We take $f(x) = A_t x$ with $A_t \in \R^{n \times n}$, so that for each admissible sensing process $u^s$ the conditional law of $x_t$ given the observation history is Gaussian by the Kalman--Bucy filter with control-dependent coefficients \cite{liptser1977statistics}.
The objective is to minimize
\begin{multline}
    J = \E\bigg[\phi(\mu_T, \Sigma_T) + \int_0^T\!\Bigl(q(\mu_t, \Sigma_t)
+ \tfrac{1}{2}(u^a_t)^\top R_a u^a_t \\
+ \rho(u^s_t)\Bigr)\,dt\bigg], 
\label{eq:cost}
\end{multline}
with $R_a \succ 0$, $\rho \geq 0$, and
\begin{align}
q(\mu, \Sigma) &:= \E_{x \sim \Normal(\mu, \Sigma)}[\bar q(x)],
\nonumber\\
\phi(\mu, \Sigma) &:= \E_{x \sim \Normal(\mu, \Sigma)}[\bar\phi(x)].
\label{eq:cost-aug}
\end{align}

The conditional distribution $\pi_t = \Normal(\mu_t, \Sigma_t)$ has covariance evolving via the controlled Riccati ODE
\begin{equation}
\dot\Sigma_t = a(t, \Sigma_t, u^s_t),
\label{eq:riccati}
\end{equation}
where
\begin{align}
a(t, \Sigma, u) &:= A_t \Sigma + \Sigma A_t^\top + Q
- D(\Sigma, u), \label{eq:a}\\
D(\Sigma, u) &:= \Sigma\,C(u)^\top R_o^{-1} C(u)\,\Sigma.
\label{eq:D}
\end{align}
Let $L(\Sigma, u)\,L(\Sigma, u)^\top = D(\Sigma, u)$.
The belief mean evolves as
\begin{equation}
d\mu_t = \bigl(f(\mu_t) + B u^a_t\bigr)\,dt
+ L(\Sigma_t, u^s_t)\,d\beta_t,
\label{eq:mu}
\end{equation}
where $\beta_t$ is a standard Brownian motion derived from the innovation $d\tilde y_t := dy_t - C(u^s_t)\mu_t\,dt$.
For fixed $C$, equations \eqref{eq:riccati}--\eqref{eq:mu} reduce to the Gaussian belief-space model of \cite{das2026cdc}, in which the Riccati ODE is autonomous and $\Sigma_t$ is a predetermined schedule.
When $C$ depends on $u^s_t$, the Riccati equation \eqref{eq:riccati} becomes non-autonomous, $\Sigma_t$ is a controlled state, and the sufficient statistic is $(\mu_t, \Sigma_t)$. The resulting closed-loop structure is illustrated in Fig. 1.
\begin{problem}[Partially observed control with controlled sensing]
\label{prob:main}
Find an admissible pair of controls $(u^a_t, u^s_t)_{t \in [0,T]}$ minimizing the cost \eqref{eq:cost} subject to the dynamics \eqref{eq:state}--\eqref{eq:obs} and the induced belief evolution \eqref{eq:riccati}--\eqref{eq:mu}.
\end{problem}
\section{Main Result}

\noindent \begin{definition}[Matching set]
\label{def:matching}
For each $\Sigma \in \Spp$,
\begin{equation}
\Us(\Sigma) := \bigl\{u \in \R^{\ell_s} :
D(\Sigma, u) = \lambda\,B R_a^{-1} B^\top\bigr\}.
\label{eq:matching-set}
\end{equation}
\end{definition}

\begin{assumption}[Measurable matching selector]
\label{ass:selector}
There exists a measurable map $\kappa: [0, T] \times \Spp \to \R^{\ell_s}$ with $\kappa(t, \Sigma) \in \Us(\Sigma)$ for all $(t, \Sigma)$.
\end{assumption}

Define
\begin{align}
a_\kappa(t, \Sigma) &:= a(t, \Sigma, \kappa(t, \Sigma)), \nonumber\\
\rho_\kappa(t, \Sigma) &:= \rho(\kappa(t, \Sigma)), \nonumber\\
D^* &:= \lambda B R_a^{-1} B^\top.
\label{eq:kappa-defs}
\end{align}

\begin{assumption}[Regularity]
\label{ass:regularity}
(i) $a_\kappa$ is locally Lipschitz in $\Sigma$ uniformly on compact time intervals, and the Cauchy problem $\dot\Sigma_t = a_\kappa(t, \Sigma_t)$ with $\Sigma_0 \in \Spp$ admits a unique solution on $[0, T]$ with $\Sigma_t \in \Spp$. (ii) $\bar q$, $\bar\phi$, and $\rho$ are bounded below. (iii) The PDE \eqref{eq:linear-pde} admits a bounded classical solution $\Psi$ on $[0, T] \times \R^n \times \Spp$ with the prescribed terminal condition.
\end{assumption}

The \emph{$\kappa$-restricted problem} fixes $u^s_t = \kappa(t, \Sigma_t)$ and optimizes only over $u^a$.
Let $V(t, \mu, \Sigma)$ denote its value function.
The corresponding HJB equation is
\begin{align}
-\partial_t V &= q + \rho_\kappa
+ (\nabla_\mu V)^\top f(\mu) \nonumber\\
&\quad + \ip{\nabla_\Sigma V}{a_\kappa}
+ \tfrac{1}{2}\tr(D^* \nabla^2_{\mu\mu} V) \nonumber\\
&\quad + \min_{u^a}\Bigl\{
\tfrac{1}{2}(u^a)^\top R_a u^a
+ (u^a)^\top B^\top \nabla_\mu V\Bigr\},
\label{eq:hjb}
\end{align}
with terminal condition $V(T, \mu, \Sigma) = \phi(\mu, \Sigma)$.
Arguments of $q$, $\rho_\kappa$, $a_\kappa$ are suppressed on the right.

\begin{theorem}[Constrained-matching Cole--Hopf]
\label{thm:main}
Under Assumptions~\ref{ass:selector}--\ref{ass:regularity}, the substitution $V = -\lambda \log \Psi$ converts \eqref{eq:hjb} into the linear PDE
\begin{equation}
\begin{aligned}
-\partial_t \Psi =\;& -\frac{q + \rho_\kappa}{\lambda}\,\Psi
+ (\nabla_\mu \Psi)^\top f \\
& + \ip{a_\kappa}{\nabla_\Sigma \Psi}
+ \tfrac{1}{2}\tr(D^* \nabla^2_{\mu\mu} \Psi),
\end{aligned}
\label{eq:linear-pde}
\end{equation}
with terminal condition $\Psi(T, \mu, \Sigma) = \exp(-\phi(\mu, \Sigma)/\lambda)$.
Moreover, $\Psi$ admits the Feynman--Kac representation
\begin{equation}
\begin{aligned}
\Psi(t, \mu, \Sigma) = \E^\kappa\!\bigg[&
\exp\!\Bigl(-\tfrac{1}{\lambda}\!\int_t^T\!(q + \rho_\kappa)\,ds
\\[-2pt]
& \;\; - \tfrac{\phi(\mu_T, \Sigma_T)}{\lambda}\Bigr)
\,\bigg|\, \mu_t = \mu,\ \Sigma_t = \Sigma\bigg],
\end{aligned}
\label{eq:fk}
\end{equation}
where under $\E^\kappa$ the augmented state evolves as
\begin{align}
d\mu_s &= f(\mu_s)\,ds + L^*\,d\beta_s, \nonumber\\
\dot\Sigma_s &= a_\kappa(s, \Sigma_s),
\label{eq:augmented}
\end{align}
with $L^* (L^*)^\top = D^*$.
\end{theorem}

\begin{proof}
The bracketed minimization in \eqref{eq:hjb} is quadratic in $u^a$ with first-order condition $R_a u^{a*} + B^\top \nabla_\mu V = 0$, giving $u^{a*} = -R_a^{-1} B^\top \nabla_\mu V$.
Substituting:
\begin{align}
&\tfrac{1}{2}(u^{a*})^\top R_a u^{a*}
+ (u^{a*})^\top B^\top \nabla_\mu V \nonumber\\
&\;= \tfrac{1}{2}(\nabla_\mu V)^\top B R_a^{-1} B^\top \nabla_\mu V
- (\nabla_\mu V)^\top B R_a^{-1} B^\top \nabla_\mu V \nonumber\\
&\;= -\tfrac{1}{2}(\nabla_\mu V)^\top B R_a^{-1} B^\top \nabla_\mu V,
\label{eq:foc-sub}
\end{align}
yielding the reduced HJB
\begin{align}
-\partial_t V
&= q + \rho_\kappa
+ (\nabla_\mu V)^\top f \nonumber\\
&\quad + \ip{\nabla_\Sigma V}{a_\kappa}
+ \tfrac{1}{2}\tr(D^* \nabla^2_{\mu\mu} V) \nonumber\\
&\quad - \tfrac{1}{2}(\nabla_\mu V)^\top B R_a^{-1} B^\top \nabla_\mu V.
\label{eq:hjb-reduced}
\end{align}
From $V = -\lambda \log \Psi$,
\begin{align}
\partial_t V &= -\tfrac{\lambda}{\Psi}\partial_t \Psi,
\label{eq:CH-1a}\\
\nabla_\mu V &= -\tfrac{\lambda}{\Psi}\nabla_\mu \Psi,
\label{eq:CH-1b}\\
\nabla_\Sigma V &= -\tfrac{\lambda}{\Psi}\nabla_\Sigma \Psi,
\label{eq:CH-1c}\\
\nabla^2_{\mu\mu} V &= -\tfrac{\lambda}{\Psi}\nabla^2_{\mu\mu}\Psi
+ \tfrac{\lambda}{\Psi^2}(\nabla_\mu \Psi)(\nabla_\mu \Psi)^\top.
\label{eq:CH-2}
\end{align}
Substituting \eqref{eq:CH-1a}--\eqref{eq:CH-2} into \eqref{eq:hjb-reduced} and using $\tr(M v v^\top) = v^\top M v$ on the cross term,
\begin{align}
\frac{\lambda}{\Psi}\partial_t \Psi
&= q + \rho_\kappa
- \frac{\lambda}{\Psi}(\nabla_\mu \Psi)^\top f
- \frac{\lambda}{\Psi}\ip{\nabla_\Sigma \Psi}{a_\kappa}
\nonumber\\
&\quad - \frac{\lambda}{2\Psi}\tr(D^* \nabla^2_{\mu\mu}\Psi)
+ \frac{\lambda}{2\Psi^2}(\nabla_\mu \Psi)^\top D^* (\nabla_\mu \Psi)
\nonumber\\
&\quad - \frac{\lambda^2}{2\Psi^2}(\nabla_\mu \Psi)^\top
B R_a^{-1} B^\top (\nabla_\mu \Psi).
\label{eq:substituted}
\end{align}
By \eqref{eq:kappa-defs}, $D^* = \lambda B R_a^{-1} B^\top$, so the two $\Psi^{-2}$ terms in \eqref{eq:substituted} combine as
\begin{equation}
\frac{\lambda}{2\Psi^2}(\nabla_\mu \Psi)^\top
\bigl[D^* - \lambda B R_a^{-1} B^\top\bigr](\nabla_\mu \Psi) = 0.
\label{eq:cancel}
\end{equation}
Multiplying the surviving terms by $-\Psi/\lambda$ yields \eqref{eq:linear-pde}; the terminal condition follows from $V(T, \cdot) = \phi$.

For \eqref{eq:fk}, the linearity of $f$ from \S\ref{sec:setup} together with Assumption~\ref{ass:regularity}(i) ensures that the augmented system \eqref{eq:augmented} admits a unique strong solution on $[t, T]$ valued in $\R^n \times \Spp$.
Define
\begin{equation}
M_s := \exp\!\Bigl(-\tfrac{1}{\lambda}\!\int_t^s
(q + \rho_\kappa)(r, \mu_r, \Sigma_r)\,dr\Bigr)
\,\Psi(s, \mu_s, \Sigma_s).
\label{eq:martingale}
\end{equation}
By It\^o's formula \cite{karatzas1991brownian}, the drift of $M_s$ equals $\partial_s\Psi + \mathcal{L}_\kappa\Psi - ((q+\rho_\kappa)/\lambda)\Psi$, where
\begin{equation}
\mathcal{L}_\kappa\Psi
= (\nabla_\mu\Psi)^\top f
+ \ip{a_\kappa}{\nabla_\Sigma\Psi}
+ \tfrac{1}{2}\tr(D^*\nabla^2_{\mu\mu}\Psi)
\label{eq:generator}
\end{equation}
is the generator of \eqref{eq:augmented} (no $\Sigma$-Hessian since $\Sigma$ is deterministic).
By \eqref{eq:linear-pde} this drift vanishes, so $M_s$ is a local martingale.
By Assumption~\ref{ass:regularity}(ii)--(iii), there exists a constant $m \in \R$ with $q + \rho_\kappa \ge m$, and $\Psi$ is uniformly bounded on $[t, T] \times \R^n \times \Spp$ by some constant $K$.
Hence
\begin{equation}
|M_s| \le K \exp\!\bigl(-\tfrac{m}{\lambda}(s - t)\bigr)
\le K \exp\!\bigl(|m|(T - t)/\lambda\bigr)
\label{eq:M-bound}
\end{equation}
uniformly on $[t, T]$, so $M_s$ is a true martingale, and taking conditional expectation at $s = T$ yields \eqref{eq:fk}.
\end{proof}

\begin{remark}[Sensor richness]
$\Us(\Sigma)$ may be empty.
If $C(u) = u\,C_0$ for a fixed $C_0 \in \R^{p \times n}$, then $D(\Sigma, u) = u^2\,\Sigma C_0^\top R_o^{-1} C_0\,\Sigma$ scales a single rank-bounded matrix, and matching $D^*$ requires $D^*$ to be proportional to $\Sigma C_0^\top R_o^{-1} C_0\,\Sigma$, which is generically restrictive when $n > 1$.
Sufficient conditions for non-emptiness on a forward-invariant subset of $\Spp$ are the natural next question.
\end{remark}

\begin{remark}[Scope]
\label{rem:scope}
Theorem~\ref{thm:main} solves the matching-constrained variant of Problem~\ref{prob:main}, in which $u^s_t$ is required to lie in $\Us(\Sigma_t)$ at every time.
The unconstrained sensing Hamiltonian
\begin{equation}
H^s(u) = \rho(u) + \ip{\nabla_\Sigma V}{a(t, \Sigma, u)}
+ \tfrac{1}{2}\tr\bigl(D(\Sigma, u)\,\nabla^2_{\mu\mu} V\bigr)
\label{eq:full-Hs}
\end{equation}
need not be minimized within $\Us(\Sigma)$, so the constrained optimum and the full optimum generally differ.
The matching constraint is treated here as a design specification arising from the requirement that path-integral linearization be available.
\end{remark}

\begin{remark}[Nonlinear extension]
\label{rem:nonlinear}
For nonlinear drift $f$, the same Cole--Hopf algebra carries over formally under an extended Kalman filter Gaussian approximation, with a second-order residual in the deviation of the mean from a nominal trajectory; a rigorous residual analysis is left to future work.
\end{remark}

\section{Scalar Example}

\noindent Consider $dx_t = u^a_t\,dt + dw_t$ with $dy_t = u^s_t\,x_t\,dt + \sigma_o\,d\nu_t$, $\rho \equiv 0$, and $n = p = \ell_a = \ell_s = 1$, $B = 1$.
Then
\begin{align}
D(\Sigma, u) &= \frac{\Sigma^2 u^2}{R_o}, \nonumber\\
a(\Sigma, u) &= 1 - \frac{\Sigma^2 u^2}{R_o}.
\label{eq:scalar-D-a}
\end{align}
The matching set is
\begin{equation}
\Us(\Sigma) = \biggl\{
\pm \frac{\sqrt{\lambda R_o}}{\Sigma\sqrt{R_a}}\biggr\},
\label{eq:scalar-set}
\end{equation}
a two-point set.
The positive selector $\kappa(\Sigma) = +\sqrt{\lambda R_o}/(\Sigma\sqrt{R_a})$ is measurable on $\Spp$ and locally Lipschitz away from $\Sigma = 0$.
With this $\kappa$,
\begin{align}
D^* &= \frac{\lambda}{R_a}, \nonumber\\
a_\kappa(\Sigma) &= 1 - \frac{\Sigma^2 \kappa(\Sigma)^2}{R_o}
= 1 - \frac{\lambda}{R_a},
\label{eq:scalar-Dstar}
\end{align}
and \eqref{eq:linear-pde} becomes
\begin{equation}
-\partial_t \Psi = -\frac{q(\mu, \Sigma)}{\lambda}\Psi
+ \Bigl(1 - \frac{\lambda}{R_a}\Bigr)\partial_\Sigma \Psi
+ \frac{\lambda}{2 R_a}\partial^2_{\mu\mu}\Psi.
\label{eq:scalar-pde}
\end{equation}
The covariance trajectory under $\kappa$ is
\begin{equation}
\Sigma_t = \Sigma_0 + \Bigl(1 - \frac{\lambda}{R_a}\Bigr)\,t,
\label{eq:scalar-sigma}
\end{equation}
and Assumption~\ref{ass:regularity}(i) is verified directly:
\begin{itemize}
\item If $\lambda \le R_a$, then $\Sigma_t \ge \Sigma_0 > 0$ for
all $t \ge 0$, the matching schedule is globally feasible, and the critical case $\lambda = R_a$ gives $\dot\Sigma_t \equiv 0$ (the covariance is frozen at $\Sigma_0$).
\item If $\lambda > R_a$, then $\Sigma_t \downarrow 0$ at
\begin{equation}
t_* = \frac{\Sigma_0}{\lambda/R_a - 1},
\label{eq:scalar-t-star}
\end{equation}
after which $\kappa(\Sigma)$ blows up and the trajectory exits $\Spp$.
The result is valid only for horizons $T < t_*$.
\end{itemize}
A bounded sensor $|u^s| \le u_{\max}$ further restricts feasibility to $\Sigma \ge \sqrt{\lambda R_o / R_a}\,/\,u_{\max}$, again a condition on the forward-invariant region of $\Spp$.

\bibliographystyle{IEEEtran}
\bibliography{references_smc}

\end{document}